\documentclass[12pt]{article}

\usepackage[utf8]{inputenc}

\usepackage{graphicx}
\usepackage{amssymb}
\usepackage{amsfonts}
\usepackage{amscd}
\usepackage{amsthm}
\usepackage{amsmath}
\usepackage{float}
\usepackage{xcolor}
\usepackage{url}

\newtheorem{theorem}{Theorem}
\newtheorem{corollary}[theorem]{Corollary}

\newtheorem{lemma}[theorem]{Lemma}
\newtheorem{proposition}[theorem]{Proposition}
\newtheorem{remark}[theorem]{Remark}

\title{Epidemics from the Eye of the Pathogen}

\author
{Faryad Darabi Sahneh$^{1}$, William Fries$^{2}$, Joseph C. Watkins$^{1,2,3,4}$, \\Joceline Lega$^{1,3,4,\ast}$\\
\\
\normalsize{$^{1}$ Department of Mathematics, University of Arizona,}\\
\normalsize{$^{2}$ Interdisciplinary Program in Applied Mathematics, University of Arizona,}\\
\normalsize{$^{3}$ Department of Epidemiology and Biostatistics, University of Arizona,}\\
\normalsize{$^{4}$ BIO5 Institute, University of Arizona,}\\
\normalsize{Tucson, AZ, 85721 USA}\\
\normalsize{$^\ast$To whom correspondence should be addressed; E-mail:  lega@math.arizona.edu.}}

\date{}

\begin{document} 
\maketitle

\begin{abstract}
 While a common trend in disease modeling is to develop models of increasing complexity, it was recently pointed out that outbreaks appear remarkably simple when viewed in the incidence vs. cumulative cases (ICC) plane. This article details the theory behind this phenomenon by analyzing the stochastic SIR (Susceptible, Infected, Recovered) model in the cumulative cases domain. We prove that the Markov chain associated with this model reduces, in the ICC plane, to a pure birth chain for the cumulative number of cases, whose limit leads to an independent increments Gaussian process that fluctuates about a deterministic ICC curve. We calculate the associated variance and quantify the additional variability due to estimating incidence over a finite period of time. We also illustrate the universality brought forth by the ICC concept on real-world data for Influenza A and for the COVID-19 outbreak in Arizona.
\end{abstract}

\section{Introduction: Outbreaks beyond the time domain and the ICC perspective}

As evidenced by the COVID-19 pandemic, societies throughout the world are highly vulnerable to disease outbreaks \cite{Morens13}. To understand the mechanism involved in disease spread and eventually provide a framework for effective public health guidance, scientists have developed numerous mathematical, statistical, and computational models of infectious disease dynamics \cite{Mckendrick1927,WALTERS20181}. But a dilemma quickly emerges: because disease spread is inherently complex, realistic descriptions commonly rely on a large number of parameters that are often unidentifiable or difficult to estimate, thereby leading to huge uncertainty in associated forecasts \cite{Edeling20}. As is typically the case with nonlinear systems, reducing the dynamics to a core nonlinear model and quantifying the associated uncertainty should provide a viable compromise between complexity and simplicity. The ICC approach \cite{Lega2016,Lega2020} introduces such a framework and, as illustrated in Figure \ref{fig:fig1}, uncovers what appears to be a generic property of outbreak data.

\begin{figure}
\centering
\includegraphics[width=\linewidth]{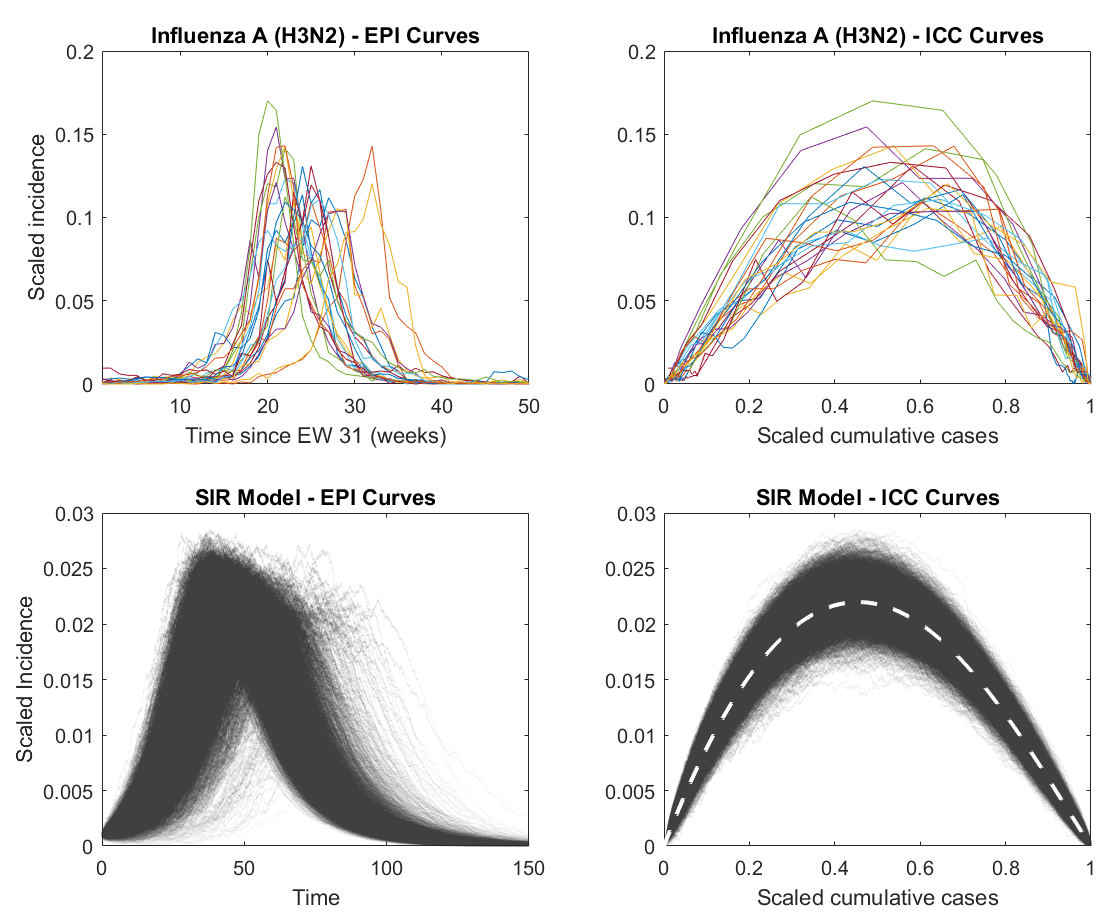}
\caption{
{\bf Top Row, left:} weekly incidence ${\mathcal I}/\hat C_\infty$ plotted as a function of time, for influenza A (H3N2) outbreaks that took place in the US between 1998 and 2019, and were of final size $\hat C_\infty > 3000$ cases. Each curve corresponds to one flu season in an HHS region. Time is measured in weeks from epidemiological week (EW) 31 of each year. The data were downloaded from the CDC Fluview database using the R {\tt cdcfluview} package \cite{CDC_fluview}. {\bf Top row, right:} the same curves plotted in the ICC plane, showing ${\mathcal I} / \hat C_\infty$ as a function of scaled cumulative cases $C / \hat C_\infty$.
{\bf Bottom row, left:} EPI curves for 5997 runs of a stochastic SIR model with size $N=2500$ and $R_0=2$. {\bf Bottom row, right:} Corresponding ICC curves, showing ${\mathcal I} / \hat C_\infty$ as a function of $C /\hat C_\infty$. The white dashed curve corresponds to Equation \ref{eq:ICC}, scaled to the expected final size $C_\infty$ of the outbreak ($C_\infty / N$ is the non-zero root of the right-hand side of Eq. \ref{eq:ICC} with $c_0$ set to 0). For the stochastic SIR model, ${\mathcal I}$ is defined as the random variable $\beta I S$ (see text for details).
\label{fig:fig1}}
\end{figure}

In most instances, the independent variable underlying the course of an epidemic is time: health authorities report numbers of new cases and deaths per day or week, forming what is commonly called an epidemiological (EPI) curve (see examples in the top left panel of Figure \ref{fig:fig1}); and modelers fit their models to this same EPI curve.  However, time – as we measure it – is not intrinsic to the spread dynamics of the pathogen. As such, focusing on temporal aspects obscure relevant properties of these
dynamics, thereby making it more difficult to fit models to data. The ICC viewpoint \cite{Lega2016,Lega2020} suggests replacing time with a monotonic, nonlinear function thereof: cumulative cases.
Therefore, in contrast to EPI curves, which describe how humans perceive outbreaks as time unfolds, ICC curves emphasize the pathogen's perspective centered on the number of people infected (i.e. the resources that have been consumed so far).
 
Figure \ref{fig:fig1} illustrates how these ideas can reveal important traits shared by different outbreaks associated with the same pathogen. The left plot of the top row shows the EPI curves of the 24 Inluenza A (H3N2) outbreaks that took place in US HHS regions between 1998 and 2019 and led to more than 3000 confirmed cases. No specific properties of these curves are readily observable, because the peak timing and peak height vary between seasons. However, when the same curves are plotted in the incidence vs.\ cumulative cases (ICC) plane, a structure emerges (top right panel), revealing similarities between each season that, as we will see, are characteristic of the disease itself. To emphasize that such properties are generic, the bottom row of Figure \ref{fig:fig1} shows similar results for multiple realizations of a stochastic SIR (Susceptible, Infected, Recovered) model via simulations on a complete graph.  Again, the universality normally hidden behind classical EPI curves (Figure \ref{fig:fig1} bottom row, left) becomes evident once time is removed from the picture and the independent variable is replaced with cumulative cases (Figure \ref{fig:fig1} bottom row, right). Incidence is defined  as $\beta I S$, which for the deterministic SIR model equals $d C / dt$. Here, $\beta$ is the microscopic contact rate of the disease, $I$ is the number of infected individuals, $S$ is the number of susceptible individuals, and $C$ is the cumulative number of cases. The parameter $\beta$ represents the probability that a given susceptible individual will encounter a {\em specific} infected individual in a population of size $n$ and therefore scales like $1/n$. Below we will introduce the population-level contact rate, $\beta_P = \beta n$, which remains finite as $n \to \infty$. For deterministic systems, an ICC curve is therefore the graph of $d C / dt$ as a function of $C$. In a discrete setting, the reported incidence is the number of new cases $\Delta C$ that occurred over a fixed period of time $\Delta$ and the incidence per unit of time is $\Delta C / \Delta$. Because disease incidence is a function of time and cumulative cases are monotonically increasing with time, the ICC curve, like the EPI curve, is {\em always} the graph of a function defined on integer values of $C$. In addition, because incidence decreases to zero between separate waves of disease spread, and consequently the cumulative cases plateau during the same periods of time, each wave of an outbreak corresponds to one "hump" (as shown in the right column of Figure \ref{fig:fig1}) of the ICC curve. One of the results of the present work is that the ICC curve of an outbreak described by the stochastic SIR model (corresponding to any of the black curves in the bottom right panel of Figure \ref{fig:fig1}) fluctuates about a mean ICC curve given by the deterministic SIR model (leading to the white dashed line in the same panel).

Dynamical systems theory has long promoted such a phase portrait perspective as displayed in the right panels of Figure \ref{fig:fig1}, since it can provide both intuitive insights and analytical approaches not easily identified under the time domain description. In \cite{Lega2016}, Lega \& Brown advocated for the relevance of this viewpoint in disease modeling; they pointed out that in many instances epidemiological data appear to follow a parabolic ICC curve, thereby suggesting that the logistic equation is a good model for the overall dynamics of $C$ as a function of time. This provided context to earlier works, in which the relevance of the logistic equation to the spread of Ebola in Africa had been noted \cite{Chowell14,Pell16}. In \cite{Lega2020}, Lega proved that the deterministic SIR compartmental model \cite{Mckendrick1927} has an exact ICC curve, whose shape is almost parabolic. The present work goes beyond the macroscopic picture provided by deterministic approaches. We analyze the statistical properties of the stochastic SIR model and explain the origins of the ICC curve from microscopic stochastic interactions.

The rest of this article is organized as follows. Section \ref{sec:stoc_SIR} introduces the stochastic SIR model and establishes that, in the limit of large populations, a single realization of this model fluctuates about the deterministic SIR ICC curve. Section \ref{sec:fluct} builds on these results to prove that the stochastic SIR model defines a Gaussian process with independent increments in the ICC plane. We quantify the associated variance, provide an elegant way of recovering a known formula for the distribution of the final size of an outbreak, find the distribution of incidence at expected disease peak, and discuss the added variability due to the difference quotient nature of the reported incidence. Section \ref{app:COVID_AZ} illustrates what some of the ideas discussed in this manuscript mean for real outbreak data. Section \ref{sec:conclusions} summarizes our results and reviews their potential applications to the analysis of outbreak data.

\section{The Stochastic SIR model and a Functional Law of Large Numbers}
\label{sec:stoc_SIR}

The SIR model consists of three compartments representing individuals susceptible of catching the disease (S), those who have the disease and are infectious (I), and those who have recovered (R) and can not longer infect others. In the stochastic version, the size of each compartment evolves according to a continuous time Markov process \cite{bartlett1949} involving the two transitions described in Table \ref{Tb:SIR}. 
Here, $n$ is the number of individuals in the population, $n_S$, $n_I$, and $n_R$ are the number of susceptible, infective, and recovered individuals respectively, and $n_C = n_I+n_R=n-n_S$ is the number of cases. The parameters $\beta$ and $\gamma$ are the individual contact and recovery rates of the disease, respectively. As mentioned above, $\beta$ scales like $1/n$.

\begin{table}[ht]
\begin{center}
\begin{tabular}{c|c|c}
event&transition&rate \\ \hline
infection&$(n_S,n_I,n_R)\to(n_S-1,n_I+1,n_R)$&$\beta n_Sn_I$ \\
recovery&$(n_S,n_I,n_R)\to(n_S,n_I-1,n_R+1)$&$\gamma n_I$ \\ \hline
\end{tabular}
\medskip
\caption{\label{Tb:SIR} \smallskip \noindent Continuous-time Markov process associated with the SIR model. The parameter $\beta$ scales like $1/n$, where $n = n_S + n_I + n_R$ is the total population size, whereas the recovery rate $\gamma$ is independent of $n$.
}
\end{center}
\end{table}

The ICC curve was developed to determine a direct  relationship between incidence and the number of cases. For the deterministic SIR model, it reads \cite{Lega2020}
\begin{equation}
    \label{eq:ICC}
\frac{dc}{dt}= \beta_P \left(c+ \frac{1}{R_0} \ln\left(1-c\right) - \frac{1}{R_0} \ln\left(1 - c_0\right)\right)
\left(1 - c\right) = G(c, c_0),
\end{equation}
where $c = n_C / n$. The {\bf population-level contact rate} $\beta_P$ is the {\bf individual-level contact rate} $\beta$ times the population size $n$. $R_0$ is the basic reproductive number, and $c_0$, the initial condition for $c$, is positive and small. Both $\beta_P$ and $R_0 = \beta_P / \gamma$ are independent of $n$ and therefore remain finite in the limit of large population sizes. The goal of this section is to prove that a relationship analogous to \eqref{eq:ICC} can be found by representing the stochastic SIR model as a multiparameter random time change (see \cite{ethier2009markov}, Section 6.2).

\subsection{Multiparameter Random Time Change Representation}
A time homogeneous pure jump Markov process on a finite state space can be represented using an appropriate number of independent rate one Poisson processes, one for each type of jump. The rate associated to any given Poisson process is random and based on the current state of the process.  Consequently, the multiparameter time change representation for the stochastic SIR model requires two Poisson processes,  $Y_i, i=1,2$, one for {\bf infection} and one for {\bf recovery}. Thus, we write the stochastic SIR model $(N_S,N_I,N_R)$ as 
\begin{eqnarray}\label{mpsir} \nonumber
N_S(t) & =& N_S(0) - Y_1\left(\int_0^t\beta N_S(u)N_I(u) du \right), \\ \nonumber
N_I(t) & =& N_I(0) + Y_1\left(\int_0^t\beta N_S(u)N_I(u) du \right)- Y_2\left(\int_0^t\gamma N_I(u)
du \right),\\
N_R(t) & =& N_R(0) +  Y_2\left(\int_0^t\gamma N_I(u) du \right).\\ \nonumber
\end{eqnarray}
As is shown in Section 6.4 of \cite{ethier2009markov}, the system of equations in (\ref{mpsir}) has a unique solution and is the SIR model introduced in Table 1. The cumulative number of cases $N_C(t)=N_I(t) + N_R(t)$ satisfies
\begin{eqnarray*}
N_C(t) &= &N_C(0) + Y_1\left(\int_0^t\beta N_S(u)N_I(u) du\right)\\ 
&= &N_C(0) + Y_1\left(\int_0^t\beta (n-N_C(u))N_I(u) du\right). 
\end{eqnarray*}
Now, taking advantage of the independent increments of the Poisson process, we may write
\begin{align}\label{cincrement}
&N_C(t+\Delta) - N_C(t) \nonumber \\
&= Y_1\left(\int_0^{t+\Delta}\beta (n-N_C(u))N_I(u) du\right) - Y_1\left(\int_0^t\beta (n-N_C(u))N_I(u) du\right) \nonumber \\
&=\tilde Y_1\left(\int_t^{t+\Delta}\beta (n-N_C(u))N_I(u) du\right),
\end{align}
where $\tilde Y_1$ is also a unit rate Poisson process. As a consequence, we have the following lemma.

\begin{lemma}\label{lemma:cdot} The rate of increase in the expected number of cases
\begin{displaymath}
\dot{C}(n_C)=\frac{d\ }{d\Delta} E[N_C(t+\Delta)-N_C(t)|N_C(t)=n_C]\Big|_{\Delta=0}
\end{displaymath}
satisfies the equation
\begin{equation}\label{cdot}
\dot{C}(n_C)= E[\beta N_I(t)(n-n_C)| N_C(t)=n_C]= \beta E[N_I(t)|N_C(t)=n_C](n-n_C).
\end{equation}
\end{lemma}

\begin{proof}
The conditional mean of the increment in  (\ref{cincrement}) is given by
\begin{align*}
&E[N_C(t+\Delta) - N_C(t)|N_C(t)=n_C]\\ 
&= E\left[\tilde Y_1\left(\int_t^{t+\Delta}\beta (n-N_C(u))N_I(u) du\right)\Big|N_C(t)=n_C\right] \\
&= \int_t^{t+\Delta}\beta E[(n-N_C(u))N_I(u)|N_C(t)=n_C]\, du
\end{align*}
Now divide by $\Delta$ and let $\Delta\to 0$.
\end{proof}
Lemma \ref{lemma:cdot} relates ${\dot C}(n_C)$ to the conditional expectation of $\beta\, n_I (n - n_C) = \beta\, n_I\, n_S$. We call the random variable ${\mathcal I} = \beta\, n_I\, n_S$ the ``macroscopic incidence.'' Our next step is to
find a formula for $E[N_I(t)|N_C(t)=n_C]$, the mean number of infective individuals given the number of cases. This relationship can be understood by examining the underlying discrete time Markov chain.

\subsection{Underlying discrete time Markov chain}
By the Doob-Gillespie algorithm \cite{Gillespie1976} and \cite{breiman1968} Section 15.6, a time-homogeneous pure-jump Markov process consists of two independent parts.
\begin{enumerate}
\item The length of time that the process remains in its current state is exponentially distributed with parameter value depending only on the current state, equal to the sum of the rates listed in the above table.
\item The jumps form an underlying time-homogeneous discrete time Markov chain.
\end{enumerate}
For the SIR model, the underlying discrete time Markov chain has two transitions, with probabilities listed in the table below.
\medskip
\begin{center}
\begin{tabular}{c|c|c}
event&transition&probability \\ \hline
infection&$(n_S,n_I,n_R)\to(n_S-1,n_I+1,n_R)$&$\beta n_S\, n_I/(\beta n_S \, n_I+\gamma n_I)$ \\
&& $ = \beta n_S/(\beta n_S+\gamma)$ \\
recovery&$(n_S,n_I,n_R)\to(n_S,n_I-1,n_R+1)$&$\gamma n_I/(\beta n_{S} \,n_I+\gamma n_I)$ \\ 
&&$ = \gamma/(\beta n_{S} +\gamma)$ \\ \hline
\end{tabular}
\end{center}
\medskip
\noindent Note that the probabilities in the last column do not depend on $n_I$ when $n_I>0$. Choosing state space variables $n_C$ and $n_I$, we recast the Markov chain transitions in terms of the total population $n$ and the number of cases $n_C$, leading to the following table.
\medskip
\begin{center}
\begin{tabular}{c|c|c}
event&transition&probability \\ \hline
infection&$(n_C,n_I)\to(n_C+1,n_I+1)$&$p(n_C)=\beta (n-n_C)/(\beta (n-n_C)+\gamma)$ \\
recovery&$(n_C,n_I)\to(n_C,n_I-1)$&$1-p(n_C)=\gamma/(\beta (n-n_C)+\gamma)$ \\ \hline
\end{tabular}
\end{center}
\medskip
\noindent Since $n$ is given, the above probabilities {\em only} depend on $n_C$, the number of cases that have occurred since the beginning of the outbreak. Using the expression for the basic reproduction number, $R_0 =  n\beta/\gamma=\beta_P/ \gamma$, we can also write
$$p(n_C) = \frac{R_0(n-n_C)/n}{R_0(n-n_C)/n+1}.$$
Consequently, we can denote the underlying Markov chain by $C_j,\ j=0,1,\ldots$ for the total number of cases at the $j$-th event. The ability to cast the Markov chain for cases alone with the number of infectives playing no role mirrors the property that the dynamics of the deterministic SIR model is completely described by a first-order differential equations for $C(t)$
\cite{Lega2020}. Note that $C_j$ is a pure birth chain with a jump up with each new infection.  
This Markov chain has a single parameter, namely $R_0$, which is a characteristic of the outbreak and independent of the population size $n$. In terms of statistical inference, the ratio that leads to the probabilities $p(n_C)$ shows that the parameter $\beta$ is {\bf ancillary} to the dynamics (see \cite{ghosh2010} for the properties of ancillary statistics). 

\subsection{The mean for the number of infected individuals}
\label{app:distributions}
We are now prepared to investigate properties of the distribution of $I_j$, the number of infected individuals at the $j$th event, when the number of cases is known. 
To this end, note that with $C_j=n_C$
\begin{equation} \label{itnc}
I_j = n_C - \left(j - n_C\right) = 2n_C-j
\end{equation}
since there have been $n_C$ infections in $j$ steps, and thus $j-n_C$ recoveries. Also note that the nature of the chain is such that $C_0=0$ and $C_1 = 1$.
Next, let
$$\tau_{n_C} = \min\{j; C_j =n_C\}$$  
denote the number of steps in the discrete Markov chain needed to reach $n_C$ cases, which is also known as a hitting time of the Markov chain. Then, $$I_{\tau_{n_C}} = 2n_C-\tau_{n_C}.$$
This shows that if we can determine the distribution of $\tau_{n_C}$, then we can also determine the distribution of $I_{\tau_{n_C}}$.
\begin{theorem}\label{limE} The expectation of $\tau_{n_C}$ satisfies
$$\lim_{n\to\infty} \frac{1}{n}E\tau_{n_C} = c-\frac{1}{R_0} \ln(1-c)$$
and consequently,
$$\lim_{n\to\infty} \frac{1}{n}EI_{\tau_{n_C}} =c+\frac{1}{R_0} \ln(1-c),$$
where $c=n_C/n$. 
\end{theorem}

\begin{proof}
A pure-birth Markov chain remains in a given state $m$ for a geometric number of steps before  making the transition to the state $m+1$. With this in mind, we can write 
\begin{equation}\label{sumsigma}
\tau_{n_C} = \sigma_1+\cdots + \sigma_{n_C-1}
\end{equation}
as the sum of independent random variables  $\sigma_m\sim Geom_1(p(m))$, where the subscript $1$ in $Geom_1(p(m))$ indicates that the the state space is $\{1,2,\ldots\}$ (rather than $\{0,1,2,\ldots\}$). Thus, $E\sigma_m = 1/p(m).$ Write 
\[
E\tau_{n_C} = \sum_{m=1}^{n_C-1} \frac{1}{p(m)} = \sum_{m=1}^{n_C-1}\frac{R_0 (n-m)/n+1}{R_0 (n-m)/n}= (n_C-1) + \frac{n}{R_0}\sum_{m=1}^{n_C-1} \frac{1}{n-m}.
\]
Then,
\begin{align*}
&\frac{1}{n}E\tau_{n_C} = c -\frac{1}{n} +\frac{1}{R_0}\sum_{m=\,1}^{n\,c-1}\frac{1}{1-m/n}\ \frac{1}{n}\\
&\to c+ \frac{1}{R_0}\int_0^{c}\frac{1}{1-q}dq = c -\frac{1}{R_0} \ln(1-c) \qquad \hbox{as}\ n\to\infty.
\end{align*}
\end{proof}

\begin{corollary}
\label{cor:mac_ICC}
The scaled rate of increase in the expected number of cases, $\dot c$ (see Lemma \ref{lemma:cdot}), satisfies
\begin{equation}\dot c = \lim_{n\to\infty} \frac{1}{n}\dot{C}([n c]) = \beta_P\left(c + \frac{1}{R_0} \ln(1-c)\right)(1-c).
\end{equation}
\end{corollary}
\begin{proof} The theorem above shows that 
$$\lim_{n\to\infty}\frac{1}{n}E[N_I(t)|N_C(t)=n\, c] = c +\frac{1}{R_0} \ln(1-c)=m_I(c),$$
where the last inequality defines $m_I(c)$. Now substitute into (\ref{cdot}) and recall that $\beta_P = n\beta$.
\end{proof}
We therefore have recovered the ICC curve  \eqref{eq:ICC} as the mean of the macroscopic incidence $\mathcal I$ in the limit as $n\to\infty$. We now turn to a description of how individual realizations of $\mathcal I$ in the stochastic SIR model fluctuate about the mean ICC curve.

\section{The statistics of fluctuations about the ICC curve}
\label{sec:fluct} 

In this section, we establish a functional central limit theorem in which the limit is an independent increments Gaussian process. 

The ingredients for a Gaussian process are a mean function and a variance-covariance function. Thus, the next task is to determine the variance structure that arises as a limit for the pure-birth Markov chain $C_j$. Recall that we set $\sigma_m\sim Geom_1(p(m))$, the number of steps that the chain remains in a given state $m$. Because the $\sigma_m$ are independent, we can use (\ref{sumsigma}) and write 
the variance of $\tau_{n_C}$ as follows.

\begin{eqnarray}\label{eq:taunc}
\hbox{Var}(\tau_{n_C}) &=& \sum_{m=1}^{n_C-1} \hbox{Var}(\sigma_m)= \sum_{m=1}^{n_C-1}\frac{1-p(m)}{p(m)^2} \nonumber \\
&=& \sum_{m=1}^{n_C-1} \frac{1/(R_0 (n-m)/n+1)}{((R_0 (n-m)/n)/(R_0 (n-m)/n+1))^2}\\
& =& \sum_{m=1}^{n_C-1}\frac{R_0 (n-m)/n+1}{R_0^2 (n-m)^2/n^2}
=\frac{n}{R_0}\sum_{m=1}^{n_C-1}\frac{1}{n-m}+\frac{n^2}{R_0^2}\sum_{m=1}^{n_C-1}\frac{1}{(n-m)^2}. \nonumber
\end{eqnarray}
Consequently, using the relationship in (\ref{itnc}).
$$\hbox{Var}(I_{\tau_{n_C}}) = \hbox{Var}(\tau_{n_C}) =\frac{n}{R_0}\sum_{m=1}^{n_C-1}\frac{1}{n-m}+\frac{n^2}{R_0^2}\sum_{m=1}^{n_C-1}\frac{1}{(n-m)^2}.$$

\begin{theorem}\label{thm:limvar}
Set $c_0=n_{C_0}/n$ and $c=n_C/n$,
\[
\lim_{n\to\infty} \frac{1}{n}(\text{Var}(I_{\tau_{n_C}})- \hbox{Var}(I_{\tau_{n_{C_0}}})) =\frac{1}{R_0} \ln\left(\frac{1-c_0}{1-c}\right) + \frac{1}{R_0^2}\frac{c-c_0}{(1-c)(1-c_0)}.
\]
\end{theorem}
\begin{proof}
Take the expression (\ref{eq:taunc}), divide by $n$ and notice that the two sums are Riemann sums. Take the limit to obtain the corresponding integral, which can be evaluated explicitly.
\end{proof}

\subsection{Functional central limit theorem}
\label{app:CLT}
We can turn the calculations above into a functional central limit theorem. To start, define
$$\bar{I}_c = \frac{1}{n}I_{\tau_{n_C}}, \qquad \bar\tau_c = \frac{1}{n}\tau_{n_C}.$$
Due to the fact that they are derived from sums of independent geometric random variables, both $\bar{I}_c$ and $\bar\tau_c$ have independent increments. In particular, set $c=n_C/n$ and define ${\cal F}_c$ to be the $\sigma$-algebra generated by $\{C_j; j \le \tau_{n_C}\}$. Then for $c_0<c_1$, $\bar\tau_{c_1}-\bar\tau_{c_0}$ and ${\cal F}_{c_0}$ are independent and by the basic properties of conditional expectation
$$E[\bar\tau_{c_1}-\bar\tau_{c_0}|{\cal F}_{c_0}]=E[\bar\tau_{c_1}-\bar\tau_{c_0}]=E\bar\tau_{c_1}-E\bar\tau_{c_0}.$$
Rearranging terms, 
\begin{equation}\label{taudiff}
E[\bar\tau_{c_1}-E\bar\tau_{c_1}|{\cal F}_{c_0}]=\bar\tau_{c_0}-E\bar\tau_{c_0},
\end{equation}
where we have used $E[E \bar \tau_{c_1} \vert {\cal F}_{c_0}] = E \bar \tau_{c_1}$ and $E[\bar \tau_{c_0} \vert {\cal F}_{c_0}] = \bar \tau_{c_0}$.
\begin{theorem}
Define
$$M^n_c = \sqrt{n}(\bar{I}_c - E\bar{I}_c)=-\sqrt{n}(\bar\tau_c - E\bar\tau_c).$$
and
$$A^n_c =  n\hbox{Var}(\bar{I}_c) = n\hbox{Var}(\bar\tau_c)=\hbox{Var}(M^n_c).$$ 
Then, $M^n_c$ and $(M^n_c)^2-A^n_c$ are mean zero martingales.
\end{theorem}
\begin{proof} The fact $E[M^n_{c_1}|{\cal F}_{c_0}]=M^n_{c_0}$ follows directory from (\ref{taudiff}), showing that $M^n_c$ is a mean zero martingale.

Using the mean zero and independent increments properties again, we find
$$E[(M^n_{c_1}-M^n_{c_0})^2|{\cal F}_{c_0}] = \hbox{Var}(M^n_{c_1}-M^n_{c_0}|{\cal F}_{c_0})=\hbox{Var}(M^n_{c_1}-M^n_{c_0})=A^n_{c_1}-A^n_{c_0}.$$
Also,
\begin{align*}
E[(M^n_{c_1}-M^n_{c_0})^2|{\cal F}_{c_0}] &=E[(M^n_{c_1})^2|{\cal F}_{c_0}]- 2M^n_{c_0}E[M^n_{c_1}|{\cal F}_{c_0}]+(M^n_{c_0})^2\\
&= E[(M^n_{c_1})^2|{\cal F}_{c_0}]- (M^n_{c_0})^2.
\end{align*}
Combining,
$$E[(M^n_{c_1})^2|{\cal F}_{c_0}]- (M^n_{c_0})^2=A^n_{c_1}-A^n_{c_0}, \quad\hbox{and}\quad  E[(M^n_{c_1})^2-A^n_{c_1}|{\cal F}_{c_0}]=(M^n_{c_0})^2-A^n_{c_0},$$
showing that 
$$(M^n_c)^2-A^n_c$$
is also a martingale. 
\end{proof}
We may therefore state the following theorem.
\begin{theorem}\label{thm:fclt}
\label{th:clm}
$M^n_c$ converges in distribution as $n\to\infty$ to a continuous independent increments Gaussian process with mean zero and variance function $\sigma^2_I(c)$.
\end{theorem}

\begin{proof}
The martingale central limit theorem has three ingredients:
\begin{enumerate}
    \item A sequence of martingales, here the sequence of stochastic processes $M^n_c$.
    \item A sequence of positive processes $A^n_c$ that compensate for $(M^n_c)^2$  so that $(M^n_c)^2-A^n_c$ is a martingale. 
    \item $A^n_c$ converges to a deterministic function continuous in $c$. Here the $A^n_c$ are themselves deterministic and converge to $\sigma^2_I(c)$ as $n\to\infty$, where
$$\sigma^2_I(c)= -\frac{1}{R_0} \ln(1-c) + \frac{1}{R_0^2}\frac{c}{1-c}.$$
We have set $c_0 = 0$ in the asymptotic expansions derived in Theorem \ref{thm:limvar} to obtain an expression in terms of $c$ only.
\end{enumerate}
Since 1, 2, and 3 hold, then the sequence of martingales converges to a mean-zero independent increments Gaussian process (see \cite{ethier2009markov} Section 7.1).
\end{proof}

\begin{remark}
\label{rm:m_I}
As a consequence of Theorem \ref{th:clm}, the mean of the scaled infected satisfies
$$E\bar{I}_c\simeq m_I(c) = c + \frac{1}{R_0} \ln (1-c)$$
and the variance 
$$n\hbox{Var}(\bar{I}_c)\simeq \sigma^2_I(c),$$
with equality in the limit as $n \to \infty$.
\end{remark}

\begin{remark} \label{rm:lln} Because $\hbox{Var}(\bar{I}_c)\to 0$ as $n\to\infty$, the convergence of expectations in Theorem \ref{limE} can, by  Theorem \ref{thm:fclt}, be replaced by convergence in mean square.
\end{remark}

\begin{remark} We can recover the number of recovered at the hitting time $\tau_{n_C}$ by noting that 
\[
R_{\tau_{n_C}}- R_{\tau_{n_0}} = (\tau_{n_C} - \tau_{n_0})-(n_C-n_0) = - (I_{\tau_{n_C}} - I_{\tau_{n_0}}) + (n_C - n_0)
\]
and thus
\[
\frac{1}{n}(R_{\tau_{n_C}}-R_{\tau_{n_0}}) =
- \frac{1}{n} (I_{\tau_{n_C}} - I_{\tau_{n_0}}) + (c - c_0) = -(\bar I_c - \bar I_{c_0}) + (c - c_0).
\]
\end{remark} 

\begin{corollary} The scaled limit of $\bar R_c= R_{\tau_{n_C}}/n$ converges to an independent increments Gaussian process. The mean of the increment from $c_0$ to $c$ is 
$$m_R(c)-m_R(c_0) =\frac{1}{R_0} \ln\left(\frac{1-c_0}{1-c}\right).$$
The variance satisfies $\sigma^2_R(c)=\sigma^2_I(c)$. The limiting processes for the scaled infective and recovered individuals have correlation $-1$.

\begin{remark} \label{rm:BM} For large $n$ and $c_0>0$, the distribution of increment $\bar I_c - \bar I_{c_0}$ can be approximated using a deterministic time change of {\bf standard Brownian motion}, $B$.
$$\bar I_c - \bar I_{c_0}\approx m_I(c)-m_I(c_0) + \frac{1}{\sqrt{n}}\big(B(\sigma_I(c)) - B(\sigma_I(c_0))\big).$$
This allow for easy and very accurate simulation of the independent increments Gaussian process.
\end{remark}
\end{corollary}

\subsection{Functional central limit theorem for the macroscopic incidence}
\label{ICLT}

We now turn to the macroscopic incidence scaled to the population size $n$, defined as
$$\frac{\mathcal I}{n} = {\mathcal I}_n = (\beta\, n) \bar{I}_c (1-c), \qquad {\mathcal I} = \beta\, n_I\, n_S,$$
where $\mathcal I$ was introduced at the end of Section \ref{sec:stoc_SIR}. Note that as $n \to \infty$, the population contact rate $\beta_P = \beta n$ remains constant for fixed $R_0 = (\beta n )/ \gamma = \beta_P/\gamma$. A central limit theorem similar to the one established in the previous section applies to ${\mathcal I}_n$. The mean scaled macroscopic incidence is obtained from the scaled number of infections 
$$m_{\mathcal I}(c) =(\beta\, n) m_I(c)(1-c),$$ 
and so is its variance, as stated below.
\begin{corollary} The scaled limit of ${\mathcal I}_n$ converges to an independent increments Gaussian process, of mean
\begin{align}
\label{eq:ICC2}
G(c,0) = G(c) &= (\beta\, n) \left(c + \frac{1}{R_0} \ln(1-c)\right) (1-c)\\ & = \beta_P \left(c + \frac{1}{R_0} \ln(1-c)\right) (1-c) \nonumber
\end{align}
and variance $\displaystyle \frac{1}{n} \sigma_{\mathcal I}^2(c)$, where
\begin{equation}
    \label{eq:var_i_macro}
\sigma_{\mathcal I}^2(c)=(\beta\, n)^2 \sigma_I^2(c) (1-c)^2= \beta_P^2 \left( -\frac{1}{R_0} \ln(1-c) + \frac{1}{R_0^2}\frac{c}{1-c}\right) (1-c)^2.
\end{equation}
\end{corollary}
The expression for $G$ in \eqref{eq:ICC2} is the same as in Equation \eqref{eq:ICC} with $c_0 / n$ set to $0$, showing agreement between the deterministic result and the mean of the stochastic model in the limit of large populations. This is the reason why we called ${\mathcal I} = \beta n_I n_S$ the macroscopic incidence. The above calculations have immediate consequences for the distribution of two quantities relevant to public health: the fraction of the population infected at peak incidence, and the final size of the outbreak. We state these results in the next section.

\subsection{Final population size and peak incidence}
\label{app:final_size}

Important properties of a disease outbreak are given at critical values $c_*$ of the fraction of cumulative cases $c = n_C / n$. Two particularly relevant examples of $c_*$ are
\begin{enumerate}
\item  $c_\wedge$, the \textit{fraction of the population that will have been infected at expected peak incidence}, i.e. when $ G'(c_\wedge)=0$, and
\item $c_\infty$, the \textit{expected final size of the outbreak}, i.e. the mean fraction of the population that will have been infected by the time the outbreak ends.
\end{enumerate}

\medskip
The first may be obtained implicitly by solving $G'(c_\wedge) = 0$ for $c_\wedge$.
\begin{align*}
0=G'(c_\wedge) &= (\beta\, n)((m_I'(c_\wedge)(1-c_\wedge) - m_I(c_\wedge))\\
&= (\beta\, n)\left(\left(1-\frac{1}{R_0}\frac{1}{1-c_\wedge}\right)(1-c_\wedge)-\left(c_\wedge +\frac{1}{R_0}\ln(1-c_\wedge)\right)\right) \\
&= (\beta\, n)\left(\left((1-c_\wedge)-\frac{1}{R_0}\right)-\left(c_\wedge + \frac{1}{R_0}\ln(1-c_\wedge)\right)\right) \\
&= (\beta\, n)\left(1-2c_\wedge-\frac{1}{R_0}(1 + \ln(1-c_\wedge)\right) \\
\Longrightarrow c_\wedge &= \frac{-1}{2R_0}(1+\ln(1-c_\wedge)) +\frac{1}{2}.
\end{align*}
The value of $c_\wedge$ may then be found numerically for specific values of $R_0$. In addition, the expression for $\sigma_{\mathcal I}(c_\wedge)$ may be applied to estimate the distribution of the scaled macroscopic incidence when $c = c_\wedge$. The bottom row of Figure \ref{fig:total-peak} shows $c_\wedge$ (left) and $\sigma_\wedge = \sigma_{\mathcal I} / (\beta n)$ (right) as functions of $R_0$, whereas Table \ref{table:two} displays their numerical values for typical values of $R_0$.

The second requires the variant of the delta method applied to hitting times (see \cite{ethier2009markov}, Section 11.4). This approach uses propagation of error to give a valuable extension of the central limit theorem. We state the result in the form of a theorem below.

\begin{theorem}
Define
$$\hat{c}_{\infty} = \inf\{c>0; \bar{I}_c = 0\}.$$
Then, $\hat c_\infty$ is approximately normally distributed, with mean $c_\infty$ such that $m_I(c_\infty)=0$ and standard deviation 
$$\sigma(\hat c_\infty) \approx \frac{1}{|m'(c_\infty)|} \frac{\sigma_I(c_\infty)}{\sqrt{n}} = \frac{\sigma_\infty}{\sqrt n}. $$
\end{theorem}

\begin{proof}
Because
$\bar{I}_c \to m_I(c)$ in $L^2$ as $n\to\infty$ and $m_I$ is continuous, we have $\hat{c}_{\infty}\to c_\infty$.
By the central limit theorem (Theorem \ref{th:clm} of the previous section), 
$$\sqrt{n}(\bar{I}_{\hat{c}_{\infty}}-m_I(\hat{c}_{\infty})) \to W,$$
where $W\sim N(0,\sigma^2_I(c_\infty))$, a normal random variable with mean 0 and variance $\sigma^2_I(c_\infty)$. Next, recall that
$m_I(c_\infty)=\bar{I}_{\hat{c}_{\infty}}=0,$
thus
$$\sqrt{n}(\bar{I}_{\hat{c}_{\infty}}-m_I(\hat{c}_{\infty}))= \sqrt{n}(m_I(c_\infty)-m_I(\hat{c}_{\infty}))
\simeq \sqrt{n}\, m_I'(c_\infty)(c_\infty-\hat{c}_{\infty}).$$
Consequently, $\hat c_\infty$ is approximately normally distributed, with mean $c_\infty$ and standard deviation 
\[
\sigma(\hat c_\infty) \simeq \frac{1}{|m'(c_\infty)|} \frac{\sigma_I(c_\infty)}{\sqrt{n}} = \frac{\sigma_\infty}{\sqrt n}.
\]
\end{proof}
Thus, the standard deviation is multiplied by a propagation of error which is inversely proportional to the slope of $m_I(c_\infty)$. The error is expanded when the slope is shallow and contracted when the slope is steep.
An expression for $c_\infty$ may be found implicitly as a function of $R_0$.
$$0=m_I(c_\infty)=c_\infty +\frac{1}{R_0}\ln (1-c_\infty), \qquad \text{i.e.} \qquad c_\infty = -\frac{1}{R_0}\ln(1-c_\infty).$$
Substituting into the variance formula, we have
$$\sigma_I^2(c_\infty) = -\frac{1}{R_0} \ln(1-c_\infty) + \frac{1}{R_0^2}\frac{c_\infty}{1-c_\infty}= c_\infty + \frac{1}{R_0^2}\frac{c_\infty}{1-c_\infty}.$$
In addition, the derivative 
$$m'_I(c_\infty)=1 -\frac{1}{R_0}\frac{1}{1-c_\infty}$$
leads to
\begin{eqnarray*}
\frac{\sigma_I^2(c_\infty)}{m'_I(c_\infty)^2} &= &\frac{c_\infty + \frac{1}{R_0^2}\frac{c_\infty}{1-c_\infty}}{\left(1 -\frac{1}{R_0}\frac{1}{1-c_\infty} \right)^2 }=\frac{R_0^2c_\infty(1-c_\infty)^2+c_\infty(1-c_\infty)}{(R_0(1-c_\infty)-1)^2}\\
&= &\frac{c_\infty(1-c_\infty)(R_0^2(1-c_\infty)+ 1)}{(R_0(1-c_\infty)-1)^2}.
\end{eqnarray*}
The square root of this expression gives $\sigma_\infty$, from which one can calculate $\sigma(\hat c_\infty)$ for specific values of $n$. The top row of Figure \ref{fig:total-peak} shows $c_\infty$ (left) and $\sigma_\infty$ (right) as functions of $R_0$. Selected numerical values are displayed in Table \ref{table:two}.

\begin{figure}
\begin{center}
\includegraphics[width=\linewidth]{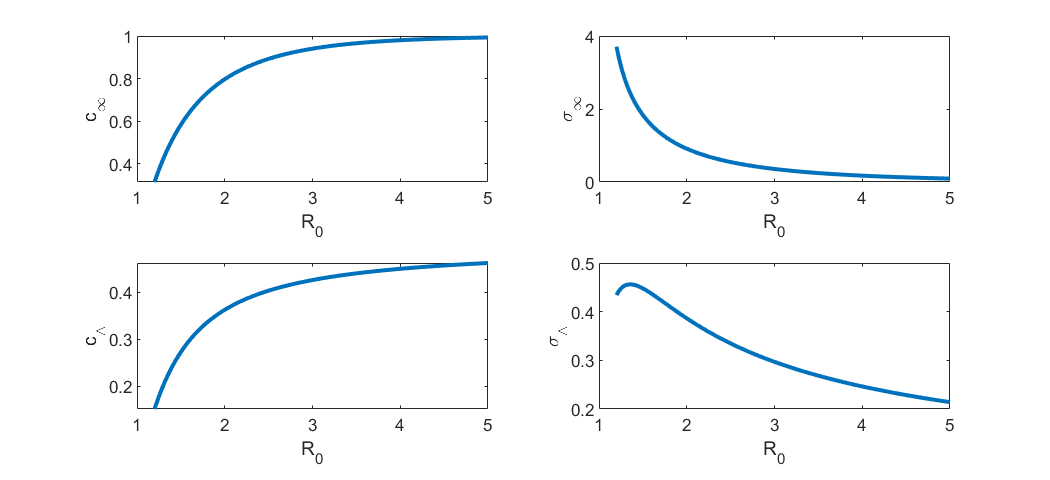}
\caption{\label{fig:total-peak} Functional dependence of select outbreak characteristics on the basic reproduction number $R_0$. {\bf Top row, left:} Mean $c_\infty$ of the fraction of population that eventually become cases, $\hat c_\infty$. {\bf Top row, right:} Behavior of $\sigma_\infty$, where $\hat c_\infty$ has standard deviation $\sigma_\infty / \sqrt{n}$. {\bf Bottom row, left:} Fraction of cumulative cases at expected peak infection $c_\wedge$. {\bf Bottom row, right:} Behavior of $\sigma_\wedge$, where ${\cal I}_n(c_\wedge)$ has standard deviation $(\beta\, n) \sigma_\wedge/\sqrt{n}$.}
\end{center}
\end{figure}

\begin{remark}
The central limit theorem for $c_\infty$ is known (see \cite{ScaliaTomba1985,ScaliaTomba1990}) but the proof presented here is new.
\end{remark}

As the graphs associated to $c_\infty$ show, the course of the pandemic looks more and more deterministic as $R_0$ grows, with an increase in cases and reduction in the standard deviation $\sigma_\infty$. The value of $c_\wedge$ increases with  $R_0$ from 0.152 to 0.462 as $R_0$ increases from 1.2 to 5.0 while the standard deviation $\sigma_\wedge$ decreases for $R_0 > 1.5$. Notably, the ratio $c_\wedge/c_\infty$ is nearly stable between 0.45 and 0.49 over a large range of values for $R_0$, reflecting the universal properties of the shape of the ICC curve.

\begin{table}
\centering
\begin{tabular}{c|cc|cc||c}
$R_0$&$\mu_{\hat c_\infty}=c_\infty$&$\sigma_\infty$&$\mu_{\hat c_\wedge}=c_\wedge$&$\sigma_\wedge$ &$c_\wedge/c_\infty$\\ \hline
1.2&0.314&3.708&0.152&0.434&0.485 \\
1.5&0.583&1.835&0.273&0.448&0.468\\
2.0&0.797&0.913&0.363&0.386&0.455\\
2.5&0.893&0.547&0.403&0.335&0.452\\
3.0&0.941&0.357&0.426& 0.297&0.453\\
3.5&0.966&0.245&0.440&0.268&0.455\\
4.0&0.980&0.174&0.450&0.246&0.459\\
4.5&0.988&0.126&0.457&0.229&0.462\\
5.0&0.993&0.094&0.462&0.214&0.465\\ \hline
\end{tabular}
\medskip
\caption{
\label{table:two}\noindent Values for the means of the fraction of population that eventually become cases $\mu_{\hat c_\infty}=c_\infty$, and the fraction of cases at peak infection $\mu_{\hat c_\wedge}=c_\wedge$. For a population of size $n$, the standard deviation for $\hat c_\infty$ is $\sigma_\infty/\sqrt{n}$. The standard deviation of ${\cal I} / n$ at $c = c_\wedge$ is $(\beta n)\, \sigma_\wedge/\sqrt{n}$. 
The final column gives the ratio of means and shows the universality of the ICC curve over a range of values for $R_0$.} 
\end{table}

\subsection{The stochastic ICC curve}

Section \ref{app:CLT} focused on the relationship between the fraction of infective individuals and the fraction of cumulative cases. This casting of the question has been shown to remove time from the analysis and with it the parameter $\beta$, the {\bf time} rate of infections.

We know bring time back into the picture by examining discrete incidence as a function of cases. Discrete, or reported, incidence ${\mathcal I}_\Delta$ is the number of new cases that occur over a given period of time $\Delta$. We shall see how the variance for ${\mathcal I}_\Delta$ depends on $\Delta$ in a nontrivial manner. To understand this dependence, we return to Equation \eqref{cincrement} and continue our analysis by computing the variance of the increment of the number of cases from time $t$ to time $t+\Delta$.
\begin{align*}
&\hbox{Var}(N_C(t+\Delta) - N_C(t)|N_C(t)=n_C) \\
&= \hbox{Var}\left(\tilde Y_1\left(\int_t^{t+\Delta}\beta (n-N_C(u))N_I(u) du\right)\Big|N_C(t)=n_C\right).
\end{align*}
where $\tilde Y_1$ is a rate-$1$ Poisson process.

To simplify notation, denote the conditional expectation   $E_{n_C} = E[\cdot|N_C(t) = n_C]$ and conditional variance  $\hbox{Var}_{n_C} = \hbox{Var}(\cdot|N_C(t) = n_C)$. Define the random variables  
$$\zeta = \int_t^{t+\Delta}\beta (n-N_C(u))N_I(u) du \quad\hbox{and}\quad \eta = \tilde Y_1(\zeta).$$
Then, $\eta \sim Pois(\zeta)$. Because the parameter in a Poisson random variable is both its mean and its variance, $E_{n_C}[\eta|\zeta] = \hbox{Var}_{n_C}(\eta|\zeta)= \zeta$. By the law of total variance,
\begin{equation}\label{eq:ltv}
\hbox{Var}_{n_C}(\eta) = E_{n_C}[\hbox{Var}_{n_C}(\eta|\zeta)] + \hbox{Var}_{n_C}( E_{n_C}[\eta|\zeta])= E_{n_C}[\zeta] + \hbox{Var}_{n_C}(\zeta).
\end{equation}
The first term of \eqref{eq:ltv} has order $\Delta$. Corollary \ref{cor:mac_ICC}  shows, that after dividing by $\Delta$, its limit as $\Delta\to 0$ is 
\begin{equation}\label{eq:var1}
   \beta (n-n_C)E[N_I(t)|N_C(t)=n_C].
\end{equation}
Expression \eqref{eq:var1} is the ICC curve. The second term is $O(\Delta^2)$. So, dividing by $\Delta^2$,
\begin{eqnarray}\label{eq:var2}
\frac{1}{\Delta^2}\hbox{Var}_{n_C}(\zeta) &=& \hbox{Var}_{n_C}\left(\frac{1}{\Delta}\int_t^{t+\Delta}\beta (n-N_C(u))N_I(u) du \right)\nonumber \\ 
&\to& \hbox{Var}_{n_C}(\beta (n-n_C) N_ I(t))= \beta^2(n-n_C)^2\hbox{Var}(N_I(t)|N_C(t)= n_C).
\end{eqnarray}
as $\Delta\to 0$. In the limit of large populations, expression \eqref{eq:var2} is the variance of the macroscopic incidence given by $\hbox{Var}({\mathcal I})=n\, \sigma_{\mathcal I}^2$, where $\sigma_{\mathcal I}^2$ is defined in Equation \eqref{eq:var_i_macro}.

Because the second term in the law of total variance is $O(\Delta^2)$, we will need to determine the second order term for $E_{n_C}[\zeta]$ to complete our analysis. To this end, we first rewrite the continuous time Markov SIR model with the number of cases $n_C$ and the number of infective individuals $n_I$  as state variables.

\begin{table}[ht]
\begin{center}
\begin{tabular}{c|c|c}
event&transition&rate \\ \hline
infection&$(n_C,n_I)\to(n_C+1,n_I+1)$&$\beta (n-n_C)n_I$ \\
recovery&$(n_C,n_I)\to(n_C,n_I-1)$&$\gamma n_I$ \\ \hline
\end{tabular} 
\end{center}
\medskip
\caption{\label{table:infrec} \smallskip \noindent Continuous-time SIR Markov process model with number of cases and number of infective as state variables.} 
\end{table}
The information in Table \ref{table:infrec} is also conveyed using the generator $G$ of the Markov process,
\begin{align*}
Gh(n_C,n_I)  = &\beta (n-n_C)n_I\big(h(n_C+1,n_I+1)-h(n_C,n_I)\big)\\
&\quad +\gamma n_I\big(h(n_C,n_I-1) -h(n_C,n_I)\big)    
\end{align*}

\begin{proposition} The $O(\Delta^2)$ term in the expansion of  $E_{n_C}[\zeta]$ is

\begin{equation}
\label{eq:var3}
\frac{1}{2}\beta^2(n-n_C)\bigg(\Big(n-n_C-1-\frac{n}{R_0}\Big)E_{n_C}[N_I(t)]
- E_{n_C}[N_I(t)^2]\bigg).
\end{equation}

\end{proposition}

\begin{proof} Set $g(n_C,n_I) = \beta(n-n_C)n_I$. Then subtract the $O(\Delta)$ term (\ref{eq:var1}) from  $E_{n_C}[\zeta]$ as defined in  Equation \eqref{eq:ltv}.
\begin{align*}
&E_{n_C}\left[\tilde Y_1\left(\int_t^{t+\Delta}g\big(N_C(u),N_I(u)\big)du\right) - g\big(n_C,N_I(t)\big)\Delta\right] \\
 &= E_{n_C}\left[ \int_t^{t+\Delta}g\big(N_C(u),N_I(u)\big)du - g\big(n_C,N_I(t)\big)\Delta\right] \\
  &= E_{n_C}\left[\int_t^{t+\Delta}\Big(g\big(N_C(u),N_I(u)\big) - g\big(n_C,N_I(t)\big)\Big) du\right] \\
  &=\int_t^{t+\Delta}E_{n_C}\Big[g\big(N_C(u),N_I(u)\big) - g\big(n_C,N_I(t)\big)\Big] du.
\end{align*}

Divide by $\Delta^2$ and take a limit using, successively, l'H\^{o}pital's rule and the definition of the generator. 
\begin{align} \label{eq:d2}
&\lim_{\Delta\to 0}\frac{1}{\Delta^2}\int_t^{t+\Delta}E_{n_C}\Big[g\big(N_C(u),N_I(u)\big) - g\big(n_C,N_I(t)\big)\Big] du \\ \nonumber
&=\lim_{\Delta\to 0}\frac{1}{2\Delta} E_{n_C}\Big[g\big(N_C(t+\Delta),N_I(t+\Delta)\big) - g\big(n_C,N_I(t)\big)\Big]  \\ \nonumber
&= \frac{1}{2}  E_{n_C}\Big[Gg\big(n_C,N_I(t)\big)\Big] =  \frac{1}{2}  E\Big[Gg\big(n_C,N_I(t)\big)\Big|N_C(t) = n_C\Big] 
\end{align}
To evaluate the generator $G$ on $g$, note that
\begin{align*}
g(n_C+1,n_I+1)-g(n_C,n_I)& =\beta(n-n_C-n_I -1)\\ 
g(n_C,n_I-1) -g(n_C,n_I) &= -\beta(n-n_C)
\end{align*}
So,
\begin{align*}
Gg(n_C,n_I) &= \beta(n-n_C)n_I\beta(n-n_C-n_I -1)-\gamma n_I\beta(n-n_C)\\
&=\beta(n-n_C)n_I\big(\beta(n-n_C-n_I -1)-\gamma\big) \\
&=\beta(n-n_C)\Big(\big(\beta(n-n_C-1)-\gamma\big)n_I -\beta n_I^2\Big) \\
&=\beta(n-n_C)\Big(\beta\big((n-n_C-1)-n\frac{1}{R_0}\big)n_I -\beta n_I^2\Big)
\end{align*}
Now, put this in the expression for the limit in Equation \eqref{eq:d2}.
\end{proof}

\begin{theorem} \label{thm:variance} The variance of the incidence over a time interval $\Delta$ is to order $\Delta^2$,
\begin{align*}
   & \frac{1}{n}\hbox{Var}(N_C(t+\Delta) - N_C(t)|N_C(t)=n_C) \\
   \simeq &\ \beta_P(1-c)m_I(c)\Delta \\
   & +\beta_P^2(1-c)\left(\frac{1}{2}\left(\left(1-c - \frac{1}{R_0}\right)m_I(c) - m_I(c)^2\right)+ (1-c)\sigma_I^2(c)\right)\Delta^2\\
   & +O(\Delta^3)
\end{align*}
as $\Delta \to 0$, with equality in the limit as $n\to\infty$.
 \end{theorem}

\begin{proof}
Recall that $\beta_P = n\beta$, $R_0 = \beta_P/\gamma$, and $c=n_C/n$. We take the three expressions (\ref{eq:var1}),  (\ref{eq:var3}), and (\ref{eq:var2}) arising from Equation \ref{eq:ltv} in order.
\medskip

\begin{enumerate}
\item $O(\Delta)$ for $E_{n_C}[\zeta]$.
\begin{align*}
&\frac{1}{n}\beta (n-n_C)E\big[N_I(t)|N_C(t)=n_C\big]\\
=&\beta_P(1-c)E\left[\frac{1}{n}N_I(t)\big|N_C(t)=n_C\right] \to \beta_P(1-c)m_I(c)
\end{align*}
as $n\to\infty$ by the proof of Corollary \ref{cor:mac_ICC}.
\item $O(\Delta^2)$ for $E_{n_C}[\zeta]$.

\begin{align*}
&\frac{1}{2n}\beta^2(n-n_C)\bigg(\Big(n-n_C-1-\frac{n}{R_0}\Big)E_{n_C}[N_I(t)]
- E_{n_C}[N_I(t)^2]\bigg)\\
= & \frac{1}{2} \beta_P^2(1-c)\bigg(\Big(1-c-\frac{1}{n}-\frac{1}{R_0}\Big)E_{n_C}[N_I(t)/n]
- E_{n_C}\big[(N_I(t)/n)^2\big]\bigg)\\
= & \frac{1}{2} \beta_P^2(1-c)\bigg(\Big(1-c-\frac{1}{n}-\frac{1}{R_0}\Big)E_{n_C}[N_I(t)/n]\\
& \qquad \qquad \qquad - \Big(\big(E_{n_C}[N_I(t)/n]\big)^2+\hbox{Var}_{n_C}\big(N_I(t)/n\big)\Big)\bigg)\\
\simeq & \frac{1}{2} \beta_P^2(1-c)\bigg(\Big(1-c-\frac{1}{n}-\frac{1}{R_0}\Big)m_I(c)
- \Big(m_I(c)^2+\frac{\sigma_I^2(c)}{n}\Big)\bigg)\\
\to & \frac{1}{2} \beta_P^2(1-c)\bigg(\Big(1-c-\frac{1}{R_0}\Big)m_I(c)
- m_I(c)^2\bigg)
\end{align*}
where the last two lines stem from Remark \ref{rm:m_I}.

\item   $O(\Delta^2)$ for $\hbox{Var}_{n_C}(\zeta)$.

\begin{align*}
\frac{1}{n}\beta^2(n-n_C)^2\hbox{Var}\big(N_I(t)|N_C(t)=n_C\big)
& = \beta_P^2(1-c)^2 n \hbox{Var}_{n_C}\big(N_I(t)/n\big)\\
& \to \beta_P^2(1-c)^2 \sigma_I^2(c),
\end{align*}
as $n \to \infty$, by Theorem \ref{th:clm}.
\end{enumerate}
\end{proof}

\begin{figure}
\centering
\includegraphics[width=\linewidth]{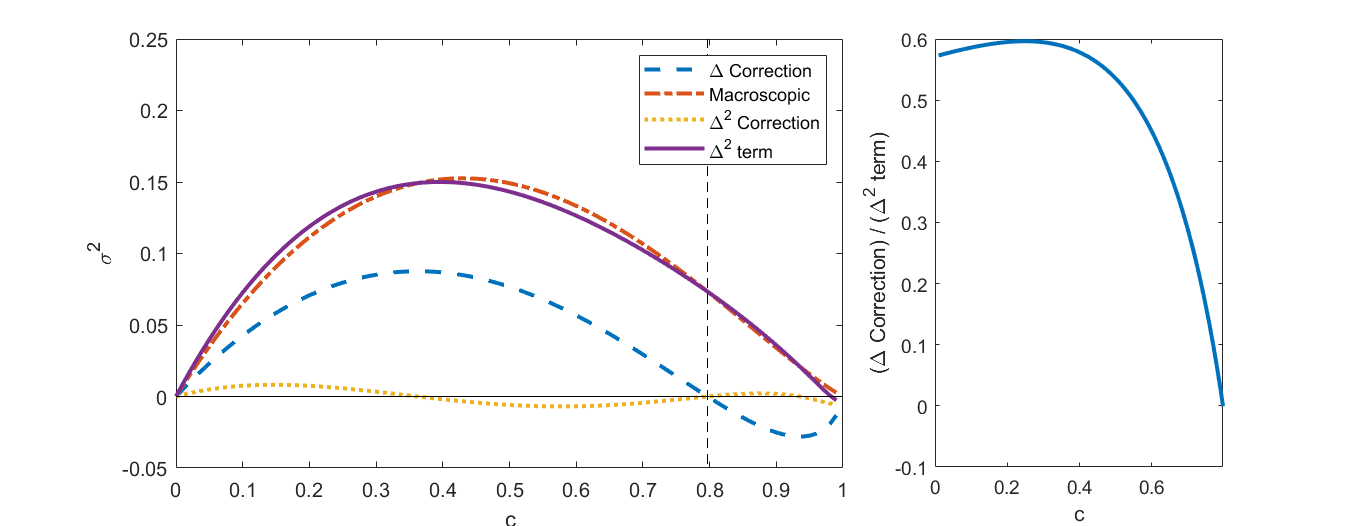}\\
\caption{Left: Graphs for the terms in the variance in the stochastic ICC curve ($R_0=2$). Dashed blue curve times $\beta_P \Delta$ is the first order term from $E_{n_C}[\zeta]$. Dash-dotted red curve times $(\beta_P \Delta)^2$ is the second order term from $\hbox{Var}_{n_C}(\zeta)$. Dotted yellow curve times $(\beta_P \Delta)^2$ is the second order term from $E_{n_C}[\zeta]$. Sum of second order terms shown in solid violet. Right: The graph times $\beta_p\Delta$ is the ratio of the second to the first order terms. }
\label{fig:iccvar}
\end{figure}

\begin{remark} Let's examine the implications for these terms.
\begin{enumerate} 
    \item The first order term in $\Delta$ for $E_{n_C}[\zeta]$ (shown in dashed blue in Figure \ref{fig:iccvar}, left) indicates that over a short time interval, the incidence is dominated by the Poisson arrival of new cases  and thus the variance is $\Delta$ times the ICC curve.
    \item The second order term in $\Delta$ arising from  $\hbox{Var}_{n_C}(\zeta)$ reflects the uncertainty in the number of infected over the time interval under consideration (shown in dash-dotted red in Figure \ref{fig:iccvar}, left). It corresponds to the variance of the macroscopic incidence ${\mathcal I}$.
    \item The second order term in $\Delta$ for $E_{n_C}[\zeta]$ (shown in dotted yellow in Figure \ref{fig:iccvar}, left) is a small perturbation of the second order term in $\hbox{Var}_{n_C}(\zeta)$.
    \item The first order term depends on $\beta_P$ and $\Delta$ through their product, the dimensionless term $\beta_P\Delta$. Correspondingly the second order terms depend on these quantities through $\beta_P^2\Delta^2$, the square of their product.
    \item The ratio of the first and second order terms (shown Figure \ref{fig:iccvar}, right with $R_0=2$) is relatively constant over a large range of values for $c$. For example, for $R_0=2$, this ratio lies between 0.5 and 0.6 for $c \in [0, 0.5]$.
    \item Thus, the first order terms dominates the variance when $\beta_P\Delta \gg \beta_P^2\Delta^2$ or for short time intervals for which $\Delta \ll 1/\beta_P$. The second order term dominates for longer time intervals when these inequalities are reversed. Both terms play a significant role for values of $\Delta$ between these two extremes.
\end{enumerate}
\end{remark}

Figure \ref{fig:varsum} summarizes these results for 20,000 Markov chain simulations, analogous to the results of the complete graph networked simulations of Figure \ref{fig:fig1}. The $\ell_2$ norm of $\sigma^2 = \hbox{Var}({\mathcal I}_\Delta / n)$, where ${\mathcal I}_\Delta = \big(N_C(t+\Delta)-N_C(t)\vert N_C(t) = n_C\big) / \Delta$, is calculated numerically and compared to the expressions shown in Theorem \ref{thm:variance} for different values of $\Delta$. This is a discrete norm since it is estimated at discrete values of $c$. Good agreement is observed for a range of values of $\Delta$, with the macroscopic term, $\hbox{Var}({\mathcal I}_n)$, becoming dominant for larger values of $\Delta$. 
\begin{figure}
\centering
\includegraphics[width=\linewidth]{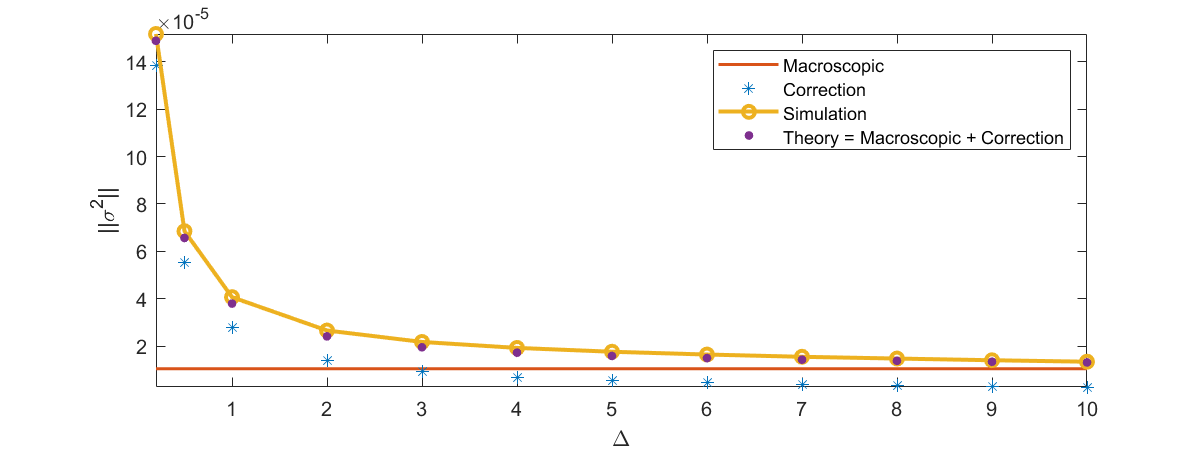}\\
\caption{Norms of the numerically evaluated variance $\sigma^2 = \hbox{Var}({\mathcal I}_\Delta / n)$ (yellow circles), of the macroscopic variance $\hbox{Var}({\mathcal I}_n)$ (solid red curve), and of the correction term to order $\Delta^2$ (blue stars), for 20,000 simulations with $N=10,000$ and different values of $\Delta$. The theoretical estimate described in Theorem \ref{thm:variance} (dots) matches the numerical simulations (yellow circles) over a broad range of values of $\Delta$.}
\label{fig:varsum}
\end{figure}

\section{Relevance of the stochastic SIR model to outbreak data}
\label{app:COVID_AZ}
The relevance of the SIR model to outbreaks is illustrated in Figure \ref{fig:fig2}, which shows the daily COVID-19 incidence in the state of Arizona for the 2020 calendar year, both in the time domain (top row: standard EPI curve) and in the cumulative cases domain (bottom row: ICC curve). The first arrow marks the end of the initial stay at home period (03/19/2020 - 05/15/2020) ordered by the Governor of Arizona \cite{EO2009,EO2018,EO2033}; the second arrow, on August 31st, indicates the end of the first six months of the outbreak (the first two cases were reported in Arizona on 03/04/2020); the third arrow marks the last day the number of cumulative cases in the state was below 300,000. Whereas the spacing between consecutive dates (108 and 83 days respectively) is similar in the time domain (top plot), this is no longer true in the cumulative case domain (bottom plot), which reveals that about twice as many cases were reported between 05/15/2020 and 08/31/2020 than between 08/31/2020 and 11/22/2020. 

\begin{figure}
\centering
\includegraphics[width=\linewidth]{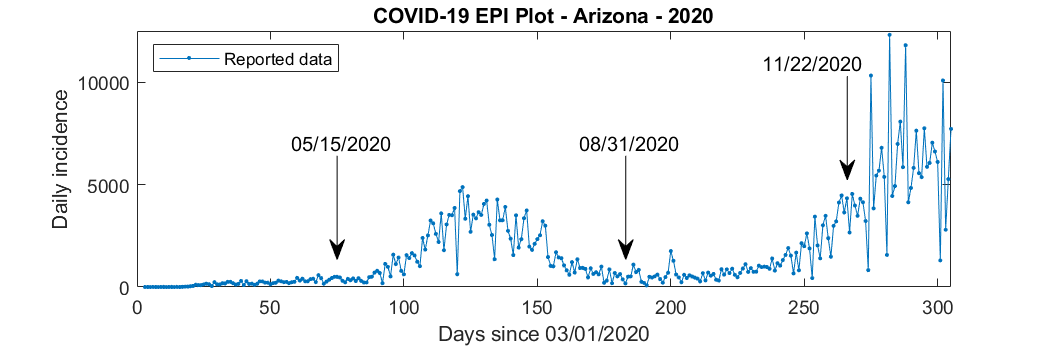}
\includegraphics[width=\linewidth]{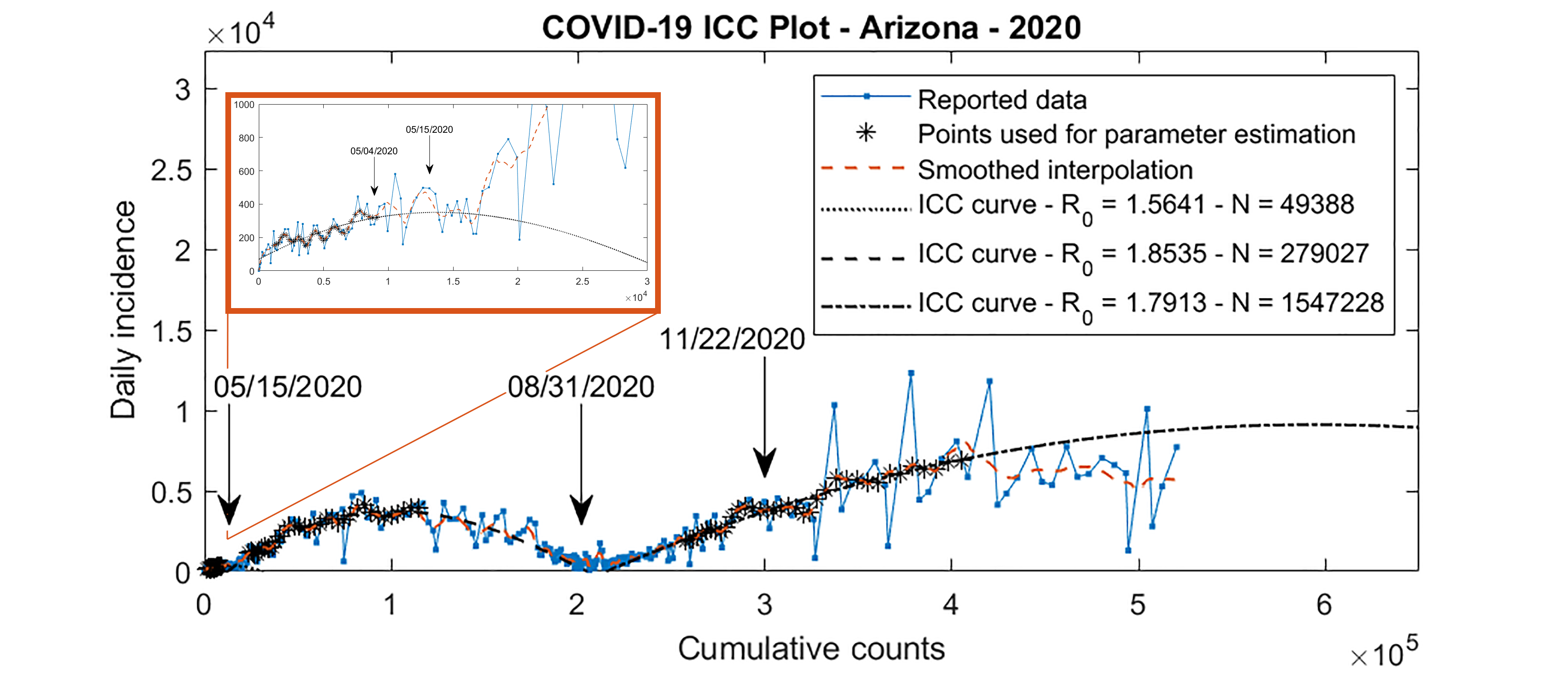}
\caption{COVID-19 outbreak in the state of Arizona in 2020, from March 1st to December 31st. {\bf Top:} Daily incidence as a function of time. {\bf Bottom:} Daily incidence as a function of cumulative cases. The inset magnifies the region with less than 30,000 cumulative cases. The first arrow corresponds to 05/04/2020, when it was announced that the stay at home order would end \cite{EO2034,EO2036} before 05/15/2020 (second arrow). The three waves are well approximated by ICC curves for the SIR model (black solid lines), whose parameters were found using a range (stars) of smoothed incidence values (yellow). The nonlinear relationship between cumulative counts $C$ and time is reflected by the change in spacing between the arrows in the top and bottom plots. COVID-19 case data provided by The COVID Tracking Project at {\it The Atlantic} under a CC BY-4.0 license \cite{CVD}.
\label{fig:fig2}}
\end{figure}

The inset displays an enlargement of the ICC curve for the first 30,000 cases (in the time domain, from 03/04/2020 to 06/10/2020). Three different waves are visible in the bottom panel of Figure \ref{fig:fig2}, each of which is locally well approximated by an ICC curve (in black) of the form 
$\bar {\mathcal I} = N\, G(c,c_0)$, where $c = C/N$, $c_0=C_0/N$, and $G$ is defined in Equation \eqref{eq:ICC}. Recall that $\beta_P$ is the population contact rate of the disease, $R_0$ is the basic reproductive number, and $C_0$ represents initial conditions. In addition, $N$ should be thought of as an effective population size. The parameters used to fit each wave vary, indicating an increase in the effective size $N$ (estimated at 49,388 individuals for the first wave, 279,027 for the second, and 1,547,228 for the third) as the outbreak unfolds, while the basic reproduction number $R_0$ fluctuates between 1.5 and 2 (respective estimates are 1.56, 1.85, and 1.79). The corresponding values of $\beta_P$ and $\gamma = \beta_P / R_0$ are $(\beta_P, \gamma) \simeq (0.12, 0.08)$, $(0.21, 0.11)$, and $(0.16, 0.09)$ respectively. 

Figure \ref{fig:fig2} suggests that each wave of the COVID-19 outbreak in Arizona is, in trend, well captured by the deterministic SIR model: the black curves, of equation $\bar {\mathcal I} = N\, G(c,c_0)$ where $G$ is defined in \eqref{eq:ICC}, are the exact relationship between incidence $\bar{\mathcal I}$ and cumulative cases $C$ for the deterministic SIR model \cite{Lega2020}. In addition, consistent with the results of this manuscript for the stochastic SIR model, each of the three waves appears to be independent from the others, and the daily incidence ${\mathcal I}_\Delta$, $\Delta = 1$, fluctuates about one of the three mean ICC curves.

\section{Conclusions}
\label{sec:conclusions}

Although not surprising from a dynamical systems point of view, the ICC perspective \cite{Lega2016,Lega2020} presents a fundamentally new way of thinking about epidemics. This article develops the corresponding theory for stochastic outbreaks and explains how they relate to deterministic ICC curves. The analysis is done for the stochastic SIR model, which captures the basic tenets of disease spread. We prove that, in the limit of large populations, the dynamics of this model in the ICC plane results from a Gaussian process with independent increments, whose distribution is concentrated about the deterministic ICC curve \eqref{eq:ICC}. The variance of ${\mathcal I}_\Delta$, the incidence over a period of time $\Delta$, is equal to the variance of the macroscopic incidence $\mathcal I$ plus a correction term that depends on $\Delta$, as described in Theorem \ref{thm:variance}. In addition, the relevance of the ICC approach becomes apparent in the nature of the dynamics: the Markov chain and its limit involve a single parameter $R_0$, and the contact rate $\beta_P$ for infections is an ancillary parameter. Both $R_0$ and $\beta_P$ are independent of the population size. In other words, shifting from the human time-centric perspective (in terms of EPI curves) to the pathogen's resource-centric perspective (in terms of ICC curves), isolates ancillary parameters from the statistical analysis of single outbreaks. 

The ability to describe outbreaks as realizations of a Gaussian process with independent increments presents many advantages. First, any outbreak can easily be simulated in the ICC plane as a deterministic time change of Brownian motion, as suggested by Remark \ref{rm:BM}. The discrete equivalent consists in looking at the current number of cumulative cases $C(t)$, drawing the new number of cases ${\mathcal I}_\Delta$ from the appropriate Gaussian distribution, adding this number to $C(t)$, and repeating these steps until no new infection occurs. Second, parameter estimation is simplified: likelihoods naturally factorize into a product of normal densities, leading to a weighted least-square minimization problem in the ICC domain. This is much simpler than the typical MCMC methods used for parameter estimation in the time domain. In addition, Fisher information can be computed explicitly to give confidence regions for model parameters, in contrast to computationally intensive simulation-based approaches. Third, the property of independent increments guarantees that estimates do not depend on the past history of the epidemic, thereby making it possible, in the case of evolving outbreaks, to infer time-dependent parameters from local data in the ICC plane. 

Although the stochastic SIR model provides a simplified description of contagion, we show in Section \ref{app:COVID_AZ} that in the ICC plane, COVID-19 incidence data fluctuate about a finite number of mean ICC curves, each having the same functional form as $G(c)$, obtained from the SIR model. Each of these mean ICC curves corresponds to one wave of the pandemic. We use Arizona as an example, but similar behaviors are observed in other states and other countries. Moreover, the independent increment nature of the process is dramatically illustrated
by these data (see Figure \ref{fig:fig2}). Estimates of $R_0$ and $N$ are entirely informed by the local dynamics of the portion of the epidemic under a given ICC curve. Data associated to the other ICC curves cannot and do not play any role.

The present analysis also shows that ICC curves can address recent challenges raised in the literature regarding time-based analysis of epidemics. In 2020, Juul {\it et al.} \cite{Juul2020} reported on the issues associated with fixed time statistics and the underestimation of extremes in epidemic curve ensembles. ICC curves circumvent many of the shortcomings of fixed time statistics because the stochastic ICC process has independent increments and thus obviates the issues of long-term correlations. In addition, the call for ``curve based'' statistics made in \cite{Juul2020} is integral to the characterization of the epidemic as a realization of a Gaussian process. This makes it possible to incorporate the entire ICC curve in the likelihood associated with any estimation, including for parameter inference, or for detecting the impact of changes -- for instance in people's behavior or due to the introduction of a vaccine, and for forecasting.

In summary, the probabilistic analysis described in the present article equips us with more powerful approaches to understand epidemic dynamics. With a change of perspective from the human to the pathogen, this article shows that the nearly century-old Kermack-McKendrick \cite{Mckendrick1927} mathematical model is again the foundation for modern, even more powerful, analytical tools that yield clearer insights into the nature of an outbreak.
\color{black}

\subsection*{Acknowledgements} We are grateful to Mohammad Javad Latifi Jebelli for insightful conversations about this work. 

\subsection*{Author contributions} FDS and JL conceived of the project. JCW led in deriving the mathematical results. JL and WF coordinated the simulations and numerical results and generated figures. All authors contributed to the writing of the manuscript and approved the final version.

\subsection*{Competing interests} All authors declare that they have no competing interests.

\end{document}